\documentclass[10pt,conference]{IEEEtran}
\usepackage{amsmath,amssymb,mathrsfs}
\usepackage{epsfig,epsf,subfigure,graphicx,graphics}
\usepackage{url}

\newtheorem{lemma}{Lemma}[section]
\newtheorem{theorem}[lemma]{Theorem}

\newtheorem{remark}[lemma]{Remark}
\newtheorem{claim}[lemma]{Claim}

\newcommand{\SNR}{\mathsf{SNR}}
\newcommand{\INR}{\mathsf{INR}}
\newcommand{\C}{\mathsf{C}^\mathsf{B}}
\newcommand{\E}{\mathrm{E}}
\newcommand{\Var}{\mathrm{Var}}

\newcommand{\etal}{{\it et al.}}

\newcommand{\aaaa}{\mathrm{(a)}}
\newcommand{\bbbb}{\mathrm{(b)}}
\newcommand{\cccc}{\mathrm{(c)}}

\newcommand{\lp}{\left(}
\newcommand{\rp}{\right)}
\newcommand{\lb}{\left[}
\newcommand{\rb}{\right]}
\newcommand{\lbp}{\left\{}
\newcommand{\rbp}{\right\}}

\newcommand{\mcal}{\mathcal}

\newcommand{\mfrak}{\mathfrak}
\newcommand{\what}{\widehat}

\newcommand{\mb}{\mathbf}
\newcommand{\mbb}{\mathbb}

\newcommand{\ra}{\rightarrow}

\graphicspath{{fig/}}

\title{Distributed Interference Cancellation in Multiple Access Channel with Transmitter Cooperation}
\author{
\authorblockN{I-Hsiang Wang}
\authorblockA{Wireless Foundations\\
University of California at Berkeley,\\
Berkeley, California 94720, USA\\
\textsf{ihsiang@eecs.berkeley.edu}}
}

\begin{document}
\maketitle

\begin{abstract}
We consider a two-user Gaussian multiple access channel with two independent additive white Gaussian interferences. Each interference is known to exactly one transmitter non-causally. Transmitters are allowed to cooperate through finite-capacity links. The capacity region is characterized to within $3$ and $1.5$ bits for the stronger user and the weaker user respectively, regardless of channel parameters. As a by-product, we characterize the capacity region of the case without cooperation to within $1$ and $0.5$ bits for the stronger user and the weaker user respectively. These results are based on a layered modulo-lattice transmission architecture which realizes distributed interference cancellation.
\end{abstract}

\section{Introduction}
In modern wireless communication systems, interference has become the major barrier for efficient utilization of available spectrum. 
In many scenarios, interferences are originated from sources close to transmitters and hence can be inferred by intelligent transmitters, while receivers cannot due to physical limitations. 
With the knowledge of interference as side information, transmitters are able to encode their information against interferences and mitigate them, even though the receiver cannot distinguish interferences from the desired signal. The simplest information theoretic model for studying such interference mitigation is the single-user point-to-point dirty-paper channel \cite{Costa_83}, which is a special case of state-dependent memoryless channels with the state known non-causaully to the transmitter \cite{GelfandPinsker_80}. It is shown that the effect of interference can be completely removed in the additive white Gaussian noise (AWGN) channel when the interference is also additive white Gaussian \cite{Costa_83}. 
As for multi-user scenarios, it has been found that when perfect state information is available non-causally at all transmitters, the capacity region of the AWGN multiple access channel (MAC) is not affected by the additive white Gaussian interference \cite{GelfandPinsker_84} \cite{KimSutivong_04}. When the sate information is known \emph{partially} to different transmitters in the MAC, however, the capacity loss caused by the interference is unbounded as the signal-to-noise ratios increase \cite{PhilosofZamir_07} \cite{Somekh-BaruchShamai_08}. Since each transmitter only has partial knowledge about the interference, interference cancellation has to be realized in a \emph{distributed} manner. 

In this paper, we consider an $K$-user Gaussian MAC with $K$ independent additive white Gaussian interferences. Each interference is known to exactly one transmitter non-causally. Furthermore, we allow transmitters to cooperate through finite-capacity links, so that transmitters can cooperatively transmit their messages and/or mitigate the known interferences. This is exactly the same model studied in \cite{PhilosofZamir_07} except for the transmitter cooperation. For simplicity, we mainly focus on the two-user case, termed as \emph{doubly-dirty MAC} in \cite{PhilosofZamir_07}. The model is depicted in Fig.~\ref{fig_Model}, where
\begin{align*}
y = x_1+x_2+s_1+s_2+z,
\end{align*}
and $z\sim\mcal{N}\lp 0, N_o\rp$ is the AWGN noise. Interference $s_i\sim\mcal{N}\lp 0,Q_i\rp$, $i=1,2$, independent of everything else, is known non-causally to transmitter $i$ \emph{only}. Power constraint at transmitter $i$ is $P_i$, $i=1,2$. Define channel parameters $\SNR_i := P_i / N_o$, $\INR_i := Q_i / N_o$, for $i=1,2$. Transmitter cooperation is induced by two orthogonal noise-free links with capacity $\C_{12}$ and $\C_{21}$, which carry signals $t_{12}$ and $t_{21}$ respectively. User $i$'s rate is denoted by $R_i$, $i=1,2$. Throughout this paper, without loss of generality we assume that user 1 has a stronger transmission power, that is, $P_1 \ge P_2$. 

\begin{figure}[htbp]
{\center
\includegraphics[width=3in]{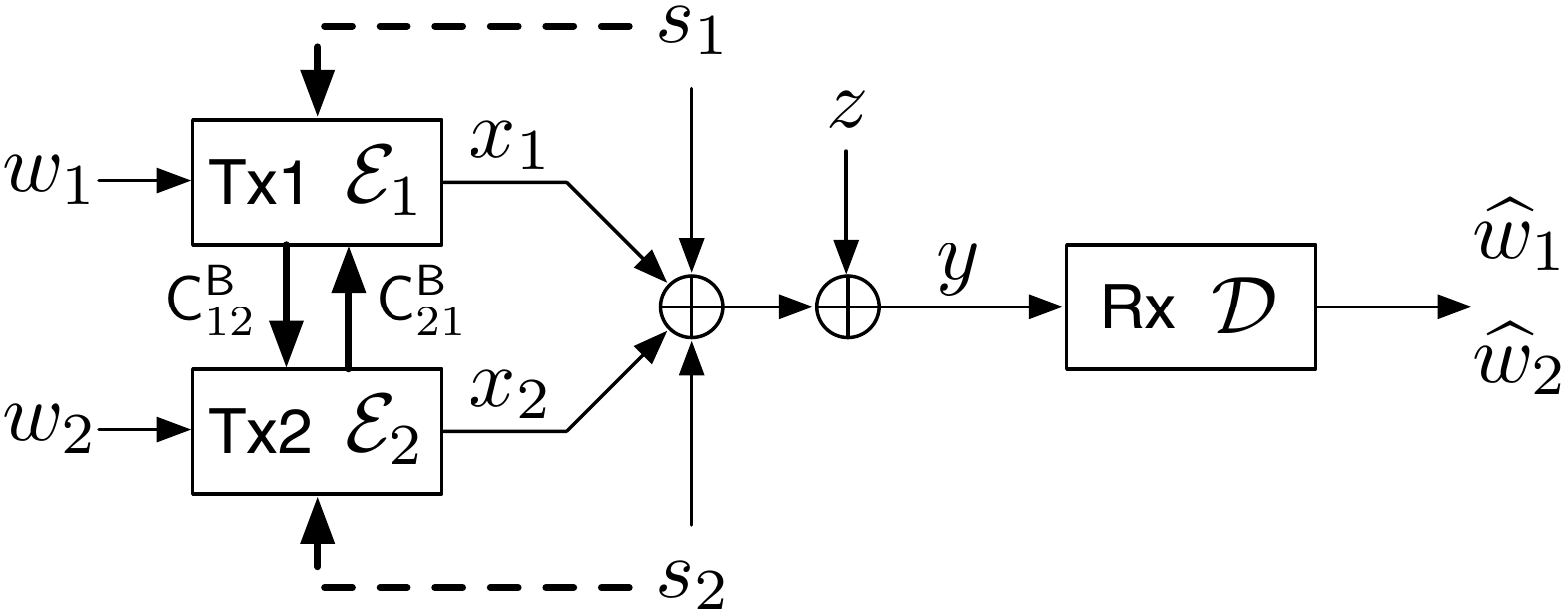}
\caption{Channel Model}
\label{fig_Model}
}
\end{figure}

State-dependent networks with partial state knowledge available at different nodes have been studied in various scenarios. Kotagiri \etal \cite{KotagiriLaneman_08} study the state-dependent MAC with state non-causally known to only one transmitter, and for the Gaussian case they characterize the capacity asymptotically at infinite interference ($Q_1=\infty, Q_2=0, \C_{12}=\C_{21}=0$) as the informed transmitter's power grows to infinity. Somekh-Baruch \etal \cite{Somekh-BaruchShamai_08} study the problem with the same set-up as \cite{KotagiriLaneman_08} while the informed transmitter knows the other's message ($Q_2=0, \C_{12}=0, \C_{21}=R_2$), and they characterize the capacity region completely. Zaidi \etal \cite{ZaidiKotagiri_09} study another case of degraded message set ($Q_2=0, \C_{21}=0, \C_{12}=R_1$). The achievability part of \cite{KotagiriLaneman_08}, \cite{Somekh-BaruchShamai_08}, and \cite{ZaidiKotagiri_09} are based on random binning. Philosof \etal \cite{PhilosofZamir_07}, on the other hand, characterize the capacity region of the doubly-dirty MAC to within a constant gap at infinite interferences (i.e., $Q_1=Q_2=\infty, \C_{12}=\C_{21}=0$), by lattice strategies. They also show that strategies based on Gaussian random binning is unboundedly worse than lattice-based strategies.
On the other hand, transmitter cooperation has also been widely investigated in various scenarios, and a non-exhaustive list includes MAC \cite{Willems_83} \cite{BrossLapidoth_08}, interference channels \cite{PrabhakaranViswanathSRC_09} \cite{WangTse_10}, MAC with state known to all transmitters \cite{BrossLapidoth_08}, and MAC with partial state known to transmitters and full state knowledge at the receiver \cite{PermuterShamai_10}.

\subsection*{Main Contribution}\label{subsec_Main}
We characterize the capacity region of the doubly-dirty MAC with transmitter cooperation to within $3$ and $1.5$ bits for $R_1$ and $R_2$ respectively. As a by-product, we characterize the capacity region of the doubly-dirty MAC without cooperation to within $1$ and $0.5$ bits for $R_1$ and $R_2$ respectively, and hence extend the constant-gap-to-optimality result in \cite{PhilosofZamir_07} to \emph{arbitrary} interference powers. The constant gap is independent of channel parameters and provides a strong guarantee on the performance. It turns out that the approximate capacity region does not depend on $\C_{12}$. The cooperation link from the stronger transmitter (Tx1) to the weaker one (Tx2) is not required to achieve the constant-gap-to-optimality performance, and it only yields a power gain which is upper bounded by a constant that does not increase with $\C_{12}$. In particular, consider the system with $\C_{21}=0$ and $\C_{12}=\infty$. As a consequence of the main results, the capacity region of this system is at most $2$ bits per user larger than the capacity region of the system without cooperation.

For the achievability part, we propose a layered superposition strategy consisting of three layers - (1) the \emph{lattice strategy} layer $\mfrak{L}$, (2) the \emph{cooperation} layer $\mfrak{C}$, and (3) the \emph{Gaussian random binning}/relaying layer $\mfrak{R}$. The hierarchy of the layers is $\mfrak{R}\ra \mfrak{C}\ra \mfrak{L}$, from the top to the bottom. Each layer treats the signals sent at higher layers as \emph{interference}, each of which is known non-causally to exactly one transmitter but not both.
In layer $\mfrak{L}$, we use a similar lattice strategy as \cite{PhilosofZamir_07} to realize distributed interference cancellation. Tx2 uses up its own power in this layer. In layer $\mfrak{C}$, the weaker transmitter Tx2 compresses the precoded information (precoded against intereference $s_2$) at a proper distortion, and uses part of the cooperation capacity to send the compression index to Tx1. Then Tx1 precodes it along with user 1's information against the aggregate interference at this layer. In layer $\mfrak{R}$, Tx2 uses the rest of the cooperation capacity to send additional data to Tx1. Tx1 uses the rest of its power to further transmit its own information or relay user 2's information, precoded against $s_1$ using either Gaussian random binning \cite{Costa_83} or lattice strategies \cite{ErezShamai_05}\footnote{The name ``Gaussian random binning layer $\mfrak{R}$" is to stress that Gaussian random binning and lattice strategies are equally good.}. For the outer bound, we use a similar argument as \cite{PhilosofZamir_09}.


\subsection*{Notations}
Throughout the paper, the block coding length is denoted by $N$, a sequence of random variables $x[1],\ldots, x[N]$ is denoted by $x^N$ and boldface $\mb{x}$ interchangeably. Logarithms are of base $2$ if not specified. We use a short-hand notation $\mcal{C}\lp \cdot \rp$ to denote $\frac{1}{2}\log\lp 1+ \cdot \rp$, $\lp\cdot\rp^+$ to denote $\max\lbp0, \cdot\rbp$, and $\log^+\lp\cdot\rp$ to denote $\lp\log\lp\cdot\rp\rp^+$. $\mbb{I}\lbp A\rbp$ denotes the indicator function, which is evaluated to $1$ if event $A$ is true and $0$ otherwise.

\section{Without Cooperation}
To better convey the idea of achievability and outer bounds, we first deal with the case without cooperation. Outer bounds are derived first. Then we describe the scheme that achieves the capacity to within a constant gap for arbitrary interference variances and transmit powers. 

\subsection{Main Result}
The main result in this section is summarized in the following lemmas and theorem.

\begin{lemma}[Outer Bounds Without Cooperation]\label{lem_OutBdwoCoop}
If nonnegative $(R_1,R_2)$ is achievable, it satisfies the following:
\begin{align}
R_1+R_2 &\le \mcal{C}\lp \SNR_1+\SNR_2\rp \label{eq_OuterSum1}\\
R_1+R_2 &\le \mcal{C}\lp \frac{1+\SNR_1+\SNR_2}{\INR_2}\rp + \mcal{C}\lp\SNR_2\rp \label{eq_OuterSum2}\\
R_2 &\le \mcal{C}\lp \SNR_2\rp \label{eq_OuterR2}.
\end{align}
\end{lemma}
\begin{proof}
See Appendix \ref{app_PfOutBdwoCoop}.
\end{proof}

\begin{lemma}[Achievable Rate Without Cooperation]\label{lem_InnBdwoCoop}
If nonnegative $\lp R_1,R_2\rp$ satisfies the following, it is achievable:
\begin{align}
R_1+R_2 &\le \frac{1}{2}\log^+\lp \frac{1}{2} + \SNR_2\rp + \mcal{C}\lp \frac{\SNR_1 - \SNR_2}{1+2\SNR_2+\INR_2}\rp \label{eq_InnerSum}\\
R_2 &\le \frac{1}{2}\log^+\lp \frac{1}{2} + \SNR_2\rp \label{eq_InnerR2}.
\end{align}
\end{lemma}
\begin{proof}
Achievability will be detailed in this section.
\end{proof}

\begin{theorem}[Constant Gap to Optimality]\label{thm_ConstGapwoCoop}
The above inner and outer bounds are within $\lp1,0.5\rp$ bits for $\lp R_1,R_2\rp$.
\end{theorem}
\begin{proof}
We combine Lemma \ref{lem_OutBdwoCoop} and \ref{lem_InnBdwoCoop} to obtain the theorem. See Appendix \ref{app_PfConstGapwoCoop} for detailed gap analysis.
\end{proof}

\subsection{Achievability}
We use the rest of this section to establish the achievability result in Lemma \ref{lem_InnBdwoCoop}. We refer to \cite{PhilosofZamir_07} and the references therein for preliminary on lattices.

The scheme consists of two layers: layer $\mfrak{R}$ and layer $\mfrak{L}$. As described in Section \ref{subsec_Main}, $\mfrak{R}$ stands for Gaussian random binning and $\mfrak{L}$ stands for lattice-based strategy. We decompose message $w_1$ into $(w_{1\mfrak{R}}, w_{1\mfrak{L}})$ and rename $w_2$ as $w_{2\mfrak{L}}$. We split the encoder at Tx1, $\mcal{E}_1$, into $\lp \mcal{E}_{1\mfrak{R}}, \mcal{E}_{1\mfrak{L}}\rp$, split the decoder at Rx, $\mcal{D}$, into $\lp \mcal{D}_{\mfrak{R}}, \mcal{D}_{\mfrak{L}}\rp$, and rename the encoder at Tx2, $\mcal{E}_2$, as $\mcal{E}_{2\mfrak{L}}$. Encoders $\mcal{E}_{1\mfrak{R}}$, $\mcal{E}_{1\mfrak{L}}$, and $\mcal{E}_{2\mfrak{L}}$ output signals $\mb{x}_{1\mfrak{R}}$, $\mb{x}_{1\mfrak{L}}$, and $\mb{x}_{2\mfrak{L}}$ respectively. Tx1 sends out the superposition of $\mb{x}_{1\mfrak{R}}$ and $\mb{x}_{1\mfrak{L}}$. Hence the receive signal can be written as
\begin{align*}
\mb{y} = \mb{x}_{1\mfrak{R}} + \mb{x}_{1\mfrak{L}} + \mb{x}_{2\mfrak{L}} + \mb{s}_1 + \mb{s}_2 + \mb{z}.
\end{align*}

{\flushleft \it Encoding}\par
{\flushleft 1) Layer $\mfrak{L}$}:
Encoders $\mcal{E}_{1\mfrak{L}}$ and $\mcal{E}_{2\mfrak{L}}$ use a lattice $\Lambda_{\mfrak{L}}$ with second moment $\Theta_{\mfrak{L}} = P_2$ and basic Voronoi region $\mcal{V}_{\mfrak{L}}$ to modulate $w_{1\mfrak{L}}$ and $w_{2\mfrak{L}}$. Generate random independent codebooks of sizes $2^{NR_{1\mfrak{L}}}$ and $2^{NR_{2\mfrak{L}}}$ according to $\mathrm{Unif}\lp\mcal{V}_{\mfrak{L}}\rp$ for $w_{1\mfrak{L}}$ and $w_{2\mfrak{L}}$ respectively. Let the codewords be $\mb{v}_{1\mfrak{L}}$ and $\mb{v}_{2\mfrak{L}}$ respectively.

Signals $\mb{x}_{1\mfrak{L}}$ and $\mb{x}_{2\mfrak{L}}$ are generated using the following modulo-lattice operation:
\begin{align}
\mb{x}_{i\mfrak{L}} = \lb \mb{v}_{i\mfrak{L}} - \alpha_{\mfrak{L}}\mb{s}_{i\mfrak{L}} - \mb{d}_{i\mfrak{L}}\rb \bmod\Lambda_{\mfrak{L}},\ i=1,2,
\end{align}
where $\mb{d}_{1\mfrak{L}}$ and $\mb{d}_{2\mfrak{L}}$, randomly and independently generated according to $\mathrm{Unif}\lp\mcal{V}_{\mfrak{L}}\rp$, are dithers known to the receiver \cite{ErezShamai_05}. $\alpha_{\mfrak{L}}$ is the MMSE coefficient $\frac{2\Theta_{\mfrak{L}}}{2\Theta_{\mfrak{L}}+N_o} = \frac{2\SNR_2}{1+2\SNR_2}$. 

$\mb{s}_{1\mfrak{L}}$ and $\mb{s}_{2\mfrak{L}}$ denote the effective \emph{interferences} known to Tx1 and Tx2 respectively in this layer: $\mb{s}_{1\mfrak{L}} = \mb{s}_1 + \mb{x}_{1\mfrak{R}}$ and $\mb{s}_{2\mfrak{L}} = \mb{s}_2$. Note that $\mb{s}_{1\mfrak{L}}$ can be produced by the higher layer encoder $\mcal{E}_{1\mfrak{R}}$.

{\flushleft 2) Layer $\mfrak{R}$}:
Layer $\mfrak{R}$ is only used at Tx1 for user 1. Encoder $\mcal{E}_{1\mfrak{R}}$ uses power $\Theta_{\mfrak{R}} = P_1-\Theta_{\mfrak{L}}$ to encode message $w_{1\mfrak{R}}$, using dirty-paper coding against interference $\mb{s}_1$. 

The encoder architecture at Tx1 is depicted in Fig.~\ref{fig_ENCwoCoop}.
\begin{figure}[htbp]
{\center
\includegraphics[width=2in]{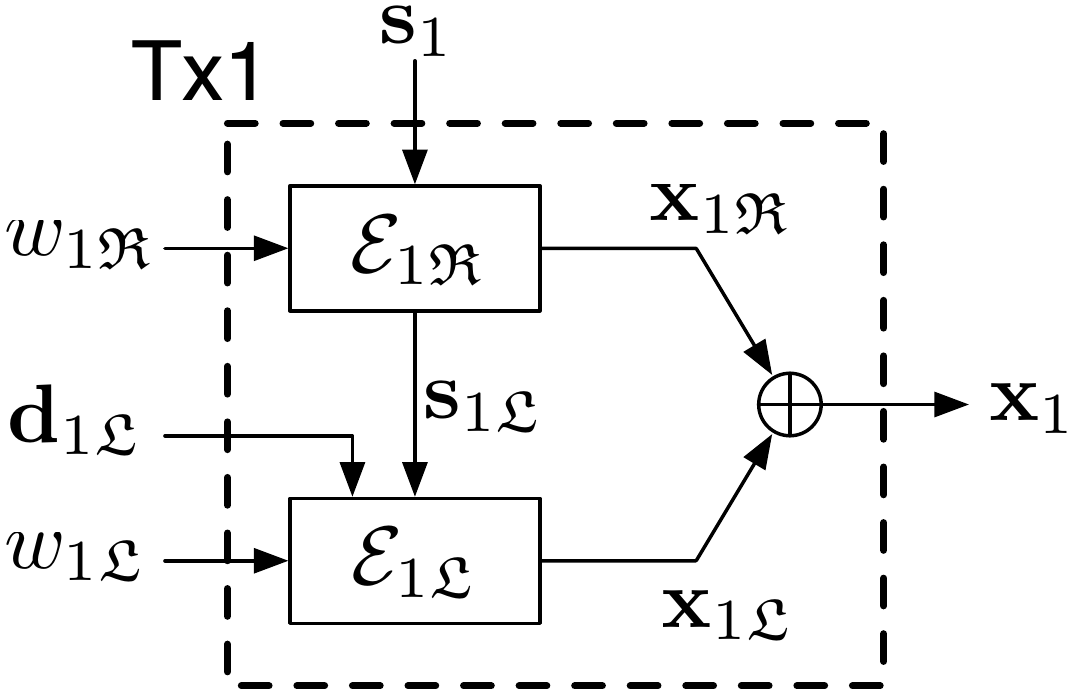}
\caption{Encoder Architecture at Tx1 (without Cooperation)}
\label{fig_ENCwoCoop}
}
\end{figure}

{\flushleft \it Decoding}\par
{\flushleft 1) Layer $\mfrak{L}$}:
Decoder $\mcal{D}_{\mfrak{L}}$ takes the input $\mb{y}$, rewritten as
\begin{align*}
\mb{y} 
&= \mb{x}_{1\mfrak{L}} + \mb{x}_{2\mfrak{L}} + \mb{s}_{1\mfrak{L}} + \mb{s}_{2\mfrak{L}} + \mb{z}_{\mfrak{L}},
\end{align*}
where $\{\mb{s}_{i\mfrak{L}},\ i=1,2\}$ are defined as above and $\mb{z}_{\mfrak{L}} = \mb{z}$.

We follow the same decoding procedure and the same line of analysis as \cite{PhilosofZamir_07} and \cite{ErezShamai_05}. The main idea is to transform the original channel into a modulo-lattice MAC, by first multiplying $\mb{y}$ by $\alpha_{\mfrak{L}}$, adding the dithers back, and taking modulo $\Lambda_{\mfrak{L}}$. The resulting output
\begin{align*}
\mb{y}_{\mfrak{L}} &= \lb \mb{y} - \lp1-\alpha_{\mfrak{L}}\rp \mb{y} + \mb{d}_{1\mfrak{L}}+\mb{d}_{2\mfrak{L}}\rb \bmod\Lambda_{\mfrak{L}}\\
&\overset{\aaaa}{=} \lb \mb{v}_{1\mfrak{L}}+ \mb{v}_{2\mfrak{L}} + \mb{z}_{\mathrm{eff},\mfrak{L}}
\rb \bmod\Lambda_{\mfrak{L}},
\end{align*} 
where $\mb{z}_{\mathrm{eff},\mfrak{L}} = \alpha_{\mfrak{L}}\mb{z}_{\mfrak{L}} - (1-\alpha_{\mfrak{L}})(\mb{x}_{1\mfrak{L}} + \mb{x}_{2\mfrak{L}})$. (a) is due to the distributive property of modulo-lattice operation. Due to dithering, $\mb{z}_{\mathrm{eff},\mfrak{L}}$ and $\mb{v}_{1\mfrak{L}}, \mb{v}_{2\mfrak{L}}$ are independent \cite{PhilosofZamir_07}. Hence the achievable rates for the modulo-lattice MAC are nonnegative $\lp R_{1\mfrak{L}}, R_{2\mfrak{L}}\rp$ satisfying \cite{PhilosofZamir_07}
\begin{align*}
R_{1\mfrak{L}}+R_{2\mfrak{L}} &\le \frac{1}{2}\log^+\lp \frac{\Theta_{\mfrak{L}}}{\alpha_{\mfrak{L}}^2N_o + (1-\alpha_{\mfrak{L}})^22\Theta_{\mfrak{L}}}\rp.
\end{align*}
This justifies the choice of $\alpha_{\mfrak{L}}$ being the MMSE coefficient $\frac{2\Theta_{\mfrak{L}}}{2\Theta_{\mfrak{L}}+N_o}$, which minimizes the effective noise variance. With this choice the achievable rates are
\begin{align}
R_{1\mfrak{L}}+R_{2\mfrak{L}} &\le \frac{1}{2}\log^+\lp \frac{1}{2} + \SNR_2\rp\label{eq_RateLattice2}.
\end{align}

{\flushleft 2) Layer $\mfrak{R}$}:
Decoder $\mcal{D}_{\mfrak{R}}$ takes the input $\mb{y}$, rewritten as
\begin{align*}
\mb{y} 
&= \mb{x}_{1\mfrak{R}} + \mb{s}_{1\mfrak{R}} + \mb{z}_{\mfrak{R}},
\end{align*}
where $\mb{s}_{1\mfrak{R}} = \mb{s}_1$ and $\mb{z}_{\mfrak{R}} = \mb{x}_{1\mfrak{L}} + \mb{x}_{2\mfrak{L}} + \mb{s}_2 + \mb{z}$.

Note that due to \emph{independent} dithering, $\{\mb{x}_{1\mfrak{L}},\mb{x}_{1\mfrak{L}},\mb{s}_1,\mb{s}_2\}$ are mutually independent. Hence the aggregate $\mb{x}_{1\mfrak{L}} + \mb{x}_{2\mfrak{L}} + \mb{s}_2$ is indeed independent of the encoding auxiliary codeword in dirty-paper coding, and its variance is $2\Theta_{\mfrak{L}}+Q_2 = 2P_2+Q_2$. We conclude that $\mb{z}_{1\mfrak{R}}$ is independent of $\mb{x}_{1\mfrak{R}}$ and $\mb{s}_{1\mfrak{R}}$, and its variance is $N_o+2P_2+Q_2$. The following claim establishes a lower bound on the achievable rate when we apply Costa's scheme to a dirty-paper channel where the additive noise is non-Gaussian:
\begin{claim}\label{claim_1}
Consider a dirty-paper channel $Y = X+S+Z$ where $X$ has power constraint $P$, $S$ is zero-mean Gaussian with variance $Q$, $\Var\lb Z\rb = N_Z$, and $\lbp X, S, Z\rbp$ are mutually independent. $S$ is known non-causally to the transmitter. Then the achievable rate is at least $\mcal{C}\lp P/N_Z\rp$.
\end{claim}
\begin{proof}
We use an argument similar to the worst-case noise property of Gaussian noise in additive noise channel \cite{DiggaviCover_01}. See Appendix \ref{app_PfClaim1} for detail.
\end{proof}

Using the above claim, we have the achievable dirty-paper coding rate \cite{Costa_83} which is evaluated assuming $\mb{z}_{1\mfrak{R}}$ is Gaussian: $R_{1\mfrak{R}} \ge 0$ satisfying
\begin{align*}
R_{1\mfrak{R}} \le \mcal{C}\lp \frac{\Theta_{\mfrak{R}}}{N_o+2P_2+Q_2}\rp
=\mcal{C}\lp \frac{\SNR_1 - \SNR_2}{1+2\SNR_2+\INR_2}\rp.
\end{align*}



Plug in $R_1 = R_{1\mfrak{L}}+R_{1\mfrak{R}}$ and $R_2 = R_{2\mfrak{L}}$ and apply Fourier-Motzkin elimination to the above achievable rate regions, we establish the achievability result in Lemma \ref{lem_InnBdwoCoop}.

We conclude this section by two remarks.
\begin{remark}
The resultant achievable region does not have a constraint on the individual rate $R_1$. This is due to the geometric structure of the achievable region (a triangle rather than a pentagon) in the lattice layer $\mfrak{L}$. 
\end{remark}

\begin{remark}
The novelty of the proposed scheme compared with that in \cite{PhilosofZamir_07} is the additional layer $\mfrak{R}$ of user 1's code and the idea of using lattice precoding to remove layer $\mfrak{R}$ in decoder $\mcal{D}_{\mfrak{L}}$. Therefore, decoders $\mcal{D}_{\mfrak{R}}$ and $\mcal{D}_{\mfrak{L}}$ can work \emph{in parallel}. In a general setting where the number of users $K\ge 2$, the same layered architecture with $K$ layers suffices to achieve the capacity region to within a constant gap, where the constant only depends on the number of users $K$. This result will be detailed in a follow-up paper.
\end{remark}

\section{With Cooperation}
With cooperation, we shall first derive the outer bounds by a slight modification of the previous arguments, taking the transmitter cooperation into account. Then we add one more layer into the previous layered strategy, which is induced by the cooperation from Tx2 to Tx1, and show that it achieves the outer bound to within a constant gap.

\subsection{Main Result}
The main result in this section is summarized as follows.

\begin{lemma}[Outer Bounds With Cooperation]\label{lem_OutBdCoop}
If nonnegative $(R_1,R_2)$ is achievable, it satisfies the following:
\begin{align}
R_1+R_2 &\le \mcal{C}\lp \SNR_1+\SNR_2+2\sqrt{\SNR_1\SNR_2}\rp \label{eq_OutCoopSum1}\\
R_1+R_2 &\le \lbp\begin{array}{l}\mcal{C}\lp \frac{1+\SNR_1+\SNR_2+2\sqrt{\SNR_1\SNR_2}}{\INR_2}\rp\\ + \mcal{C}\lp\SNR_2\rp + \C_{21}\end{array}\rbp\label{eq_OutCoopSum2}\\
R_2 &\le \mcal{C}\lp \SNR_2\rp + \C_{21}.\label{eq_OutCoopR2}
\end{align}
\end{lemma}
\begin{proof}
See Appendix \ref{app_PfOutBdCoop}.
\end{proof}

\begin{lemma}[Achievable Rate With Cooperation]\label{lem_InnBdCoop}
If nonnegative $\lp R_1,R_2\rp$ satisfies the following, it is achievable.
\begin{align}
R_1+R_2 &\le \frac{1}{2}\log^+\lp \frac{1}{2} + \SNR_2\rp + \lp \mcal{C}\lp \frac{\Theta_{\mfrak{C}}}{N_o+2P_2}\rp-\frac{1}{2} \rp^+\notag\\
&\quad + \mcal{C}\lp \frac{P_1-\Theta_{\mfrak{C}}-P_2}{N_o + \Theta_{\mfrak{C}}+2P_2+Q_2}\rp \label{eq_InnCoopSum}\\
R_2 &\le \frac{1}{2}\log^+\lp \frac{1}{2} + \SNR_2\rp + \lp \mcal{C}\lp \frac{\Theta_{\mfrak{C}}}{N_o+2P_2}\rp-\frac{1}{2} \rp^+ \notag\\
&\quad +\lp\C_{21} - r_{21}\rp. \label{eq_InnCoopR2}
\end{align}
Here we choose
\begin{align*}
\Theta_{\mfrak{C}} = \min\lbp \lp N_o+2P_2\rp \lp2^{2\C_{21}} -2\rp^+, Q_2, P_1-P_2\rbp, 
\end{align*}
and $r_{21} = \mcal{C}\lp 1+\frac{\Theta_{\mfrak{C}}}{N_o+2P_2}\rp \mbb{I}\lbp \C_{21}\ge\frac{1}{2}\rbp$.
\end{lemma}
\begin{proof}
Achievability will be detailed in this section.
\end{proof}

\begin{theorem}[Constant Gap to Optimality]\label{thm_ConstGapCoop}
The above inner and outer bounds are within $(3,1.5)$ bits for $(R_1,R_2)$.
\end{theorem}
\begin{proof}
We combine Lemma \ref{lem_OutBdCoop} and \ref{lem_InnBdCoop} to obtain the theorem. See Appendix \ref{app_PfConstGapCoop} for detailed gap analysis.
\end{proof}

\subsection{Achievability}
We shall only make use of the link from Tx2 to Tx1, as suggested by the outer bounds. In addition to the above mentioned layers $\mfrak{R}$ and $\mfrak{L}$, due to the cooperation we introduce a third middle layer $\mfrak{C}$, which denotes \emph{cooperation}. Decompose messages $w_1$ into $\lp w_{1\mfrak{R}}, w_{1\mfrak{C}}, w_{1\mfrak{L}}\rp$ and $w_2$ into $\lp w_{2\mfrak{R}}, w_{2\mfrak{C}}, w_{2\mfrak{L}}\rp$. We keep the encoder and decoder architecture as in the case without cooperation, except that now we split the encoder at Tx1, $\mcal{E}_1$, into three sub-encoders $\lp \mcal{E}_{1\mfrak{R}}, \mcal{E}_{1\mfrak{C}}, \mcal{E}_{1\mfrak{L}}\rp$, and split the decoder at Rx, $\mcal{D}$, into three sub-decoders $\lp \mcal{D}_{\mfrak{R}}, \mcal{D}_{\mfrak{C}}, \mcal{D}_{\mfrak{L}}\rp$. Encoder $\mcal{E}_{1\mfrak{C}}$ outputs signal $\mb{x}_{1\mfrak{C}}$. Tx1 sends out the superposition of $\mb{x}_{1\mfrak{R}}$, $\mb{x}_{1\mfrak{C}}$, and $\mb{x}_{1\mfrak{L}}$. Hence the receive signal can be written as
\begin{align*}
\mb{y} = \mb{x}_{1\mfrak{R}} + \mb{x}_{1\mfrak{C}} + \mb{x}_{1\mfrak{L}} + \mb{x}_{2\mfrak{L}} + \mb{s}_1 + \mb{s}_2 + \mb{z}.
\end{align*}

{\flushleft \it Encoding}\par
{\flushleft 1) Layer $\mfrak{L}$}:
We use the same scheme as in the case without cooperation. The only difference is the effective interference known to Tx1 in this layer becomes $\mb{s}_{1\mfrak{L}} = \mb{s}_1 + \mb{x}_{1\mfrak{R}} + \mb{x}_{1\mfrak{C}}$.

{\flushleft 2) Layer $\mfrak{C}$}:
In this layer, we use a lattice $\Lambda_{\mfrak{C}}$ with second moment $\Theta_{\mfrak{C}}$ and basic Voronoi region $\mcal{V}_{\mfrak{C}}$ to modulate $w_{1\mfrak{C}}$ and $w_{2\mfrak{C}}$. Generate random independent codebooks of sizes $2^{NR_{1\mfrak{C}}}$ and $2^{NR_{2\mfrak{C}}}$ according to $\mathrm{Unif}\lp\mcal{V}_{\mfrak{C}}\rp$ for $w_{1\mfrak{C}}$ and $w_{2\mfrak{C}}$ respectively. Let the codewords be $\mb{v}_{1\mfrak{C}}$ and $\mb{v}_{2\mfrak{C}}$. Tx1 and Tx2 \emph{would} transmit $\mb{x}_{1\mfrak{C}}$ and $\mb{x}_{2\mfrak{C}}$ respectively, using the following modulo-lattice operation, if they \emph{had} enough power:
\begin{align*}
\mb{x}_{i\mfrak{C}} = \lb \mb{v}_{i\mfrak{C}} - \alpha_{\mfrak{C}}\mb{s}_{i\mfrak{C}} - \mb{d}_{i\mfrak{C}}\rb \bmod\Lambda_{\mfrak{C}},\ i=1,2,
\end{align*}
where $\mb{d}_{i\mfrak{C}}$'s are dithers, $\mb{s}_{i\mfrak{C}}$'s are effective interferences known to transmitters, and $\alpha_{\mfrak{C}}$ is the MMSE coefficient. However, since Tx2 has no power left (recall that in layer $\mfrak{L}$ Tx2 has already used up its power), user 2's precoded signal has to be transmitted by Tx1 via \emph{cooperation}. Therefore, dither $\mb{d}_{2\mfrak{C}}$ is no longer needed because in this layer the received signal is solely contributed by Tx1.

The effective interference $\mb{s}_{2\mfrak{C}} = \mb{s}_2$. Tx2 first compresses
\begin{align}\label{eq_def}
\mb{x}_{2\mfrak{C}} :=  \lb \mb{v}_{2\mfrak{C}} - \alpha_{\mfrak{C}}\mb{s}_{2\mfrak{C}} \rb \bmod\Lambda_{\mfrak{C}}
\end{align}
using a Gaussian vector quantizer: $\what{\mb{x}}_{2\mfrak{C}} = \mb{x}_{2\mfrak{C}} + \what{\mb{z}}$, and $\what{\mb{z}} \sim \mcal{N}(\mb{0},\Delta \mb{I}_N)$ is independent of everything else. $\Delta$ denotes the quantization distortion. Note that the rate for Tx1 to recover $\what{\mb{x}}_{2\mfrak{C}}$ reliably is upper bounded by the rate distortion function assuming $\what{\mb{x}}_{2\mfrak{C}}$ is Gaussian, since Gaussian distribution is the differential entropy maximizing distribution under power constraint. Let the rate for sending the compression index be $r_{21}$, $r_{21}\le\C_{21}$. Hence, we have the following criterion:
\begin{align}
r_{21} \ge \frac{1}{2}\log\lp 1+\frac{\Theta_{\mfrak{C}}}{\Delta}\rp. \label{eq_QRate}
\end{align}
We shall set $r_{21}$ such that the above holds with equality if $\C_{21} \ge \frac{1}{2}\log\lp 1+\frac{\Theta_{\mfrak{C}}}{\Delta}\rp$. If not, we simply drop this layer by setting $r_{21}=\Theta_{\mfrak{C}}=0$. The value of $\Delta$ will be described later in the decoding part.

Tx2 then sends the quantization point $\what{\mb{x}}_{2\mfrak{C}}$ to Tx1. Encoder $\mcal{E}_{1\mfrak{C}}$ outputs
\begin{align*}
\mb{x}_{1\mfrak{C}} = \lb \mb{v}_{1\mfrak{C}} + \what{\mb{x}}_{2\mfrak{C}} - \alpha_{\mfrak{C}}\mb{s}_{1\mfrak{C}} - \mb{d}_{1\mfrak{C}}\rb \bmod\Lambda_{\mfrak{C}},
\end{align*}
where the effective interference $\mb{s}_{1\mfrak{C}} = \mb{s}_1 + \mb{x}_{1\mfrak{R}}$. The value of $\alpha_{\mfrak{C}}$ will be described later in the decoding part.

{\flushleft 3) Layer $\mfrak{R}$}:
Layer $\mfrak{R}$ is now shared between both users. Tx2 uses the rest of the cooperation capacity $\lp\C_{21}-r_{21}\rp$ to send message $w_{2\mfrak{R}}$ to Tx1. Tx1 uses the rest of the power, that is, $P_1 - \Theta_{\mfrak{C}} - P_2$, to encode messages $\lp w_{1\mfrak{R}},w_{2\mfrak{R}}\rp$, using dirty-paper coding or lattice strategies against interference $\mb{s}_1$.

The encoder architecture at Tx1 is depicted in Fig.~\ref{fig_ENCCoop}.
\begin{figure}[htbp]
{\center
\includegraphics[width=2in]{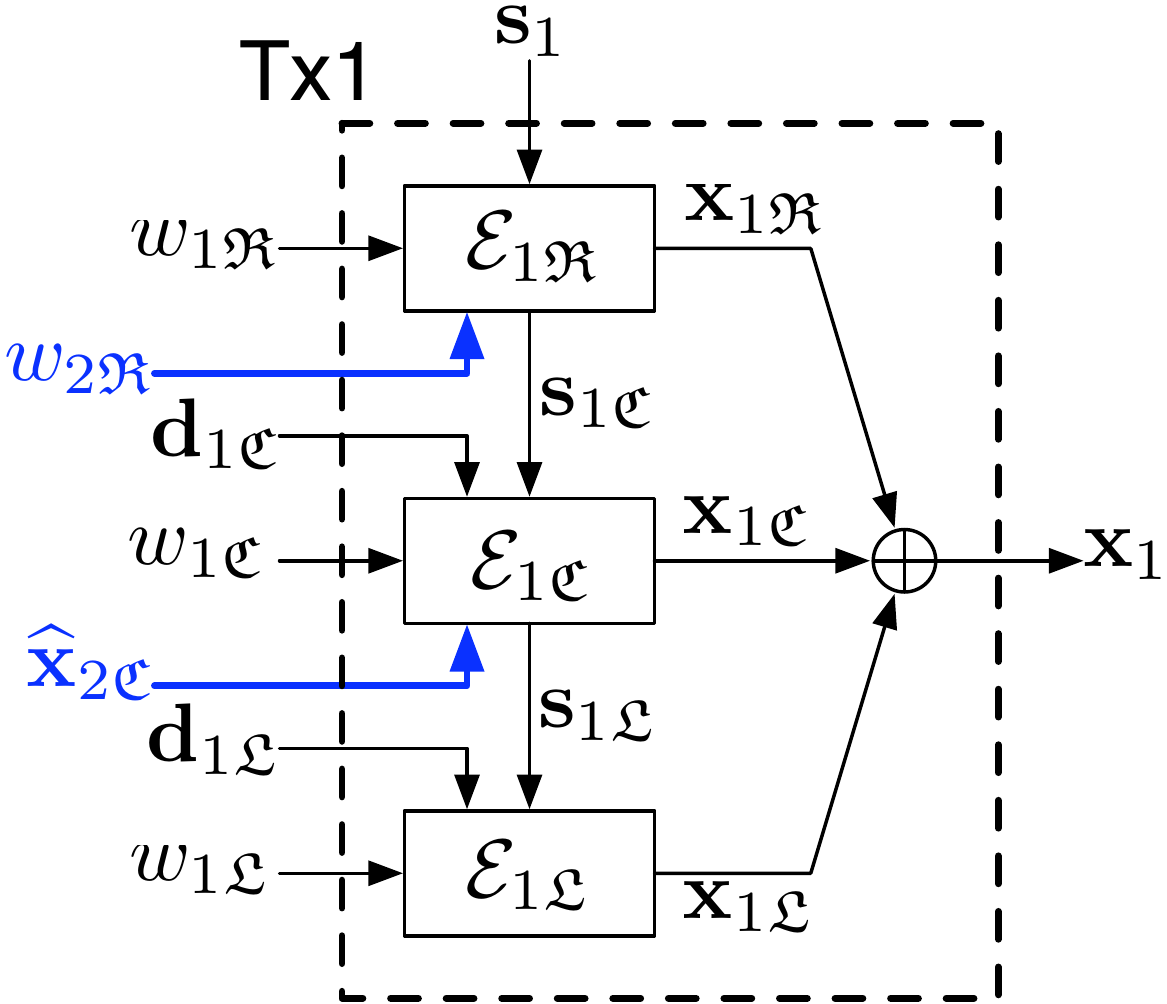}
\caption{Encoder Architecture at Tx1 (with Cooperation). Blue denotes the cooperation information}
\label{fig_ENCCoop}
}
\end{figure}

{\flushleft \it Decoding}\par
{\flushleft 1) Layer $\mfrak{L}$}:
Decoder $\mcal{D}_{\mfrak{L}}$ remains the same as in the case without cooperation, and achievable rates are described in \eqref{eq_RateLattice2}.

{\flushleft 2) Layer $\mfrak{C}$}:
Decoder $\mcal{D}_{\mfrak{C}}$ takes the input $\mb{y}$, rewritten as
\begin{align*}
\mb{y} = \mb{x}_{1\mfrak{C}} + \mb{s}_{1\mfrak{C}} + \mb{s}_{2\mfrak{C}} + \mb{z}_{\mfrak{C}},
\end{align*}
and employs the same decoding procedure as decoder $\mcal{D}_{\mfrak{L}}$. Here $\{\mb{s}_{i\mfrak{C}},\ i=1,2\}$ are defined as above and $\mb{z}_{\mfrak{C}} = \mb{x}_{1\mfrak{L}}+\mb{x}_{2\mfrak{L}}+\mb{z}$.

The equivalent modulo-lattice channel output is
\begin{align*}
\mb{y}_{\mfrak{C}} &= \lb \mb{y} - \lp1-\alpha_{\mfrak{C}}\rp \mb{y} + \mb{d}_{1\mfrak{C}} \rb \bmod\Lambda_{\mfrak{C}}\\
&= \lb \mb{v}_{1\mfrak{C}}+ \mb{v}_{2\mfrak{C}} + \mb{z}_{\mathrm{eff},\mfrak{C}}
\rb \bmod\Lambda_{\mfrak{C}},
\end{align*} 
where $\mb{z}_{\mathrm{eff},\mfrak{C}} = \what{\mb{z}} + \alpha_{\mfrak{C}}\mb{z}_{\mfrak{C}} - (1-\alpha_{\mfrak{C}})\mb{x}_{1\mfrak{C}}$. Hence the achievable rates are nonnegative $\lp R_{1\mfrak{C}}, R_{2\mfrak{C}}\rp$ satisfying
\begin{align*}
R_{1\mfrak{C}}+R_{2\mfrak{C}} &\le \frac{1}{2}\log^+\lp \frac{\Theta_{\mfrak{C}}}{\Delta + \alpha_{\mfrak{C}}^2(N_o+2P_2) + (1-\alpha_{\mfrak{C}})^2\Theta_{\mfrak{C}}}\rp
\end{align*}
We shall choose $\alpha_{\mfrak{C}} = \frac{\Theta_{\mfrak{C}}}{\Theta_{\mfrak{C}}+(N_o+2P_2)}$ to minimize the effective noise variance, which is now $\Delta + \frac{\Theta_{\mfrak{C}}(N_o+2P_2)}{\Theta_{\mfrak{C}}+(N_o+2P_2)}$. For simplicity we choose $\Delta = \frac{\Theta_{\mfrak{C}}(N_o+2P_2)}{\Theta_{\mfrak{C}}+(N_o+2P_2)}$, and hence the achievable rate region consists of nonnegative $\lp R_{1\mfrak{C}}, R_{2\mfrak{C}}\rp$ satisfying 
\begin{align*}
R_{1\mfrak{C}}+R_{2\mfrak{C}} &\le \lp \mcal{C}\lp \frac{\Theta_{\mfrak{C}}}{N_o+2P_2}\rp-\frac{1}{2} \rp^+.
\end{align*}

As for the values of $r_{21}$ and $\Theta_{\mfrak{C}}$, note that \eqref{eq_QRate} becomes $r_{21} \ge \mcal{C}\lp 1+\frac{\Theta_{\mfrak{C}}}{N_o+2P_2}\rp$ where the right-hand side is always lower bounded by $1/2$. Besides, in layer $\mfrak{R}$ the codeword in layer $\mfrak{C}$ is treated as noise, where the interference $s_2$ is also treated as noise. Hence, if $\C_{21}\ge 1/2$, we set $r_{21} = \mcal{C}\lp 1+\frac{\Theta_{\mfrak{C}}}{N_o+2P_2}\rp$, and $\Theta_{\mfrak{C}} = \min\lbp \lp N_o+2P_2\rp \lp2^{2\C_{21}} -2\rp, Q_2, P_1-P_2\rbp$. Otherwise, we set $r_{21} = \Theta_{\mfrak{C}} = 0$. 


{\flushleft 3) Layer $\mfrak{R}$}:
Decoder $\mcal{D}_{\mfrak{R}}$ uses the same procedure as in the case without cooperation to decode $\lp w_{1\mfrak{R}}, w_{2\mfrak{R}}\rp$. The only difference is the effective noise term $\mb{z}_{\mfrak{R}}$ becomes $\mb{x}_{1\mfrak{C}} + \mb{x}_{1\mfrak{L}} + \mb{x}_{2\mfrak{L}} + \mb{s}_2 + \mb{z}$.

Similar to the previous argument in the case without cooperation, we obtain the achievable rate in this layer: nonnegative $R_{1\mfrak{R}}, R_{2\mfrak{R}}$ satisfying
\begin{align*}
R_{1\mfrak{R}}+R_{2\mfrak{R}} &\le \mcal{C}\lp \frac{P_1-\Theta_{\mfrak{C}}-P_2}{N_o + \Theta_{\mfrak{C}}+2P_2+Q_2}\rp\\
R_{2\mfrak{R}} &\le \C_{21} - r_{21}.
\end{align*}

Plug in $R_1 = R_{1\mfrak{L}}+R_{1\mfrak{C}}+R_{1\mfrak{R}}$ and $R_2 = R_{2\mfrak{L}}+R_{2\mfrak{C}}+R_{2\mfrak{R}}$, and apply Fourier-Motzkin elimination to the above achievable rate regions, we establish Lemma \ref{lem_InnBdCoop}.

We conclude the paper by the following remark discussing how to sharpen the gap.
\begin{remark}
In the proposed scheme there are several points for future improvement. First, the cooperation link from Tx1 to Tx2 is not utilized, though it only provides a power gain. Second, Tx2 uses a suboptimal Gaussian VQ to compress $\mb{x}_{2\frak{C}}$ defined in \eqref{eq_def}, which is because it is technically simpler to handle. The quantization distortion $\Delta$ is also a heuristic choice, which can be further optimized.
\end{remark}

\section*{Acknowledgment}
The author thanks Prof. David Tse for motivating this work and Prof. Mich\`{e}le Wigger for inspiring discussions.

\bibliographystyle{ieeetr}

\newpage

\appendices
\section{Converse Proofs}
\subsection{Proof of Lemma \ref{lem_OutBdwoCoop}}\label{app_PfOutBdwoCoop}
{\flushleft 1) $R_1+R_2$ bound:}\par
If $R_1$ and $R_2$ are achievable, by Fano's inequality and data processing inequality, we have
\begin{align}
&N\lp R_1+R_2-\epsilon_N\rp \le I\lp w_1,w_2; y^N\rp \label{eq1}\\
&\overset{\aaaa}{=} I\lp w_1,w_2, s_1^N, s_2^N; y^N\rp - I\lp s_1^N, s_2^N; y^N | w_1,w_2\rp\\
&\overset{\bbbb}{=} h\lp y^N\rp - h\lp z^N\rp - h\lp s_1^N,s_2^N| w_1,w_2\rp\\&\quad + h\lp s_1^N,s_2^N| w_1,w_2,y^N\rp\\
&\overset{\cccc}{=} -N\log\lb\lp2\pi e\rp^3 N_oQ_1Q_2\rb/2\\&\quad + \underset{(*)}{\underbrace{h\lp y^N\rp + h\lp s_1^N,s_2^N| w_1,w_2,y^N\rp}}, \label{eq2}
\end{align}
where $\epsilon_N\rightarrow 0$ as $N\rightarrow \infty$. (a) is due to chain rule. (b) is due to the fact that $\lp x_1^N, x_2^N\rp$ is a function of $\lp w_1,w_2,s_1^N,s_2^N\rp$. (c) is due to the fact that $\lbp w_1,w_2,s_1^N,s_2^N\rbp$ are mutually independent.

Note that the term $(*)$ can be upper bounded as follows
\begin{align}
&h\lp y^N\rp + h\lp s_1^N,s_2^N| y^N\rp = h\lp y^N|s_1^N,s_2^N\rp + h\lp s_1^N,s_2^N\rp \notag\\
&\le \frac{N}{2}\log\lb\lp2\pi e\rp^3 \lp N_o+P_1+P_2\rp Q_1Q_2\rb. \label{eq3}
\end{align}
This gives the outer bound \eqref{eq_OuterSum1}.
	
On the other hand, $(*)$ can also be upper bounded as follows
\begin{align}
&(*) = h\lp y^N\rp + h\lp s_1^N| w_1,w_2,y^N\rp + h\lp s_2^N| w_1,w_2,y^N,s_1^N\rp \notag\\
&\overset{\aaaa}{\le} h\lp y^N\rp + h\lp s_1^N| y^N\rp + h\lp s_2^N| w_1,w_2,y^N,s_1^N\rp \notag\\
&\overset{\bbbb}{=} h\lp y^N|s_1^N\rp + h\lp s_1^N\rp + h\lp x_2^N+z^N|w_1,w_2,y^N,s_1^N\rp \notag\\
&\le h\lp y^N|s_1^N\rp + h\lp s_1^N\rp + h\lp x_2^N+z^N\rp \notag\\
&\le \frac{N}{2}\log\lb\lp2\pi e\rp^3 \lp N_o+Q_2+P_1+P_2\rp Q_1 \lp N_o+P_2\rp\rb. \label{eq4}
\end{align}
(a) is due to conditioning reduces entropy. (b) is due to chain rule and the fact that $y^N = x_1^N + x_2^N + s_1^N + s_2^N+z^N$ and the fact that $x_1^N$ is a function of $\lp w_1, s_1^N\rp$. Hence, this leads to the outer bound \eqref{eq_OuterSum2}.

{\flushleft 2) $R_2$ bound:}
Providing the state information $\lp s_1^N,s_2^N\rp$ to the decoder, we obtain the clean MAC without transmitter cooperation, and the bound \eqref{eq_OuterR2} is trivial.

\subsection{Proof of Lemma \ref{lem_OutBdCoop}}\label{app_PfOutBdCoop}
{\flushleft 1) $R_1+R_2$ bound}\par
The first part of the proof follows the same line as the case without cooperation, from \eqref{eq1} to \eqref{eq2}, and the upper bound on $(*)$ in \eqref{eq3} is replaced by 
\begin{align*}
\frac{N}{2}\log\lb\lp2\pi e\rp^3 \lp N_o+P_1+P_2+2\sqrt{P_1P_2}\rp Q_1Q_2\rb,
\end{align*}
taking the correlation between $x_1^N$ and $x_2^N$ into account. This gives the outer bound \eqref{eq_OutCoopSum1}.
	
On the other hand, the upper bound on $(*)$ in \eqref{eq4} is slightly modified as follows:
\begin{align*}
&(*)\\
&= h\lp y^N\rp + h\lp s_1^N| w_1,w_2,y^N\rp + h\lp s_2^N| w_1,w_2,y^N,s_1^N\rp\\
&\overset{\aaaa}{\le} h\lp y^N\rp + h\lp s_1^N| y^N\rp + h\lp s_2^N| w_1,w_2,y^N,s_1^N, t_{21}^N\rp\\
&\quad + I\lp s_2^N; t_{21}^N| w_1,w_2,y^N,s_1^N\rp\\
&\overset{\bbbb}{\le} h\lp y^N|s_1^N\rp + h\lp s_1^N\rp + h\lp s_2^N| w_1,w_2,y^N,s_1^N,t_{21}^N\rp\\
&\quad + H\lp t_{21}^N\rp\\
&\overset{\cccc}{=} h\lp y^N|s_1^N\rp + h\lp s_1^N\rp + h\lp x_2^N+z^N|w_1,w_2,y^N,s_1^N,t_{21}^N\rp\\
&\quad + H\lp t_{21}^N\rp\\
&\le h\lp y^N|s_1^N\rp + h\lp s_1^N\rp + h\lp x_2^N+z^N\rp + H\lp t_{21}^N\rp\\
&\le \frac{N}{2}\log\lb\begin{subarray}{l}\lp2\pi e\rp^3 \lp N_o+Q_2+P_1+P_2+2\sqrt{P_1P_2}\rp Q_1 \lp N_o+P_2\rp\end{subarray}\rb  + N\C_{21}.
\end{align*}
(a) is due to conditioning reduces entropy. (b) is due to chain rule and conditioning reduces entropy. (c) is due to the fact that $y^N = x_1^N + x_2^N + s_1^N + s_2^N+z^N$ and the fact that $x_1^N$ is a function of $\lp w_1, s_1^N,t_{21}^N\rp$. Hence, this leads to the outer bound \eqref{eq_OutCoopSum2}.

{\flushleft 2) $R_2$ bound}\par
Providing the state information $\lp s_1^N,s_2^N\rp$ to the decoder, we obtain the clean MAC with transmitter cooperation. Then the cut-set bound gives \eqref{eq_OutCoopR2}.

\section{Gap Analysis}
\subsection{Proof of Theorem \ref{thm_ConstGapwoCoop}}\label{app_PfConstGapwoCoop}
Compare the $R_2$ bounds: 
\begin{align*}
\eqref{eq_OuterR2} - \eqref{eq_InnerR2}&= \frac{1}{2}\log\lp 1+\SNR_2\rp - \frac{1}{2}\log^+\lp \frac{1}{2} + \SNR_2\rp\\
&\le \frac{1}{2}\log\lp 1+\SNR_2\rp - \frac{1}{2}\log\lp \frac{1+\SNR_2}{2} \rp = \frac{1}{2}.
\end{align*} 

Compare the $R_1+R_2$ bounds: if $\INR_2 \le 1+2\SNR_2$, $\eqref{eq_OuterSum1} - \eqref{eq_InnerSum}$ is upper bounded by
\begin{align*}
&\frac{1}{2}\log\lp 1+\SNR_1+\SNR_2\rp\\
& - \frac{1}{2}\log\lp \frac{1+2\SNR_2}{2}\rp - \frac{1}{2}\log\lp \frac{1+\SNR_1+\SNR_2}{2+4\SNR_2}\rp = 1.
\end{align*}
If $\INR_2 > 1+2\SNR_2$, $\eqref{eq_OuterSum2} - \eqref{eq_InnerSum}$ is upper bounded by
\begin{align*}
&\log\lp \frac{1+\INR_2+\SNR_1+\SNR_2}{\INR_2}\rp\\
& - \frac{1}{2}\log\lp \frac{1+\SNR_1+\SNR_2+\INR_2}{2\INR_2}\rp + \frac{1}{2} = 1.
\end{align*}

Hence the proof is complete.

\subsection{Proof of Theorem \ref{thm_ConstGapCoop}}\label{app_PfConstGapCoop}

{\flushleft \it Case $\C_{21} < \frac{1}{2}$}:\par
The inner bound \eqref{eq_InnCoopSum} becomes the same as the sum rate inner bound \eqref{eq_InnerSum} in the case without cooperation. Therefore it suffices to compare the outer bounds \eqref{eq_OuterSum1} with \eqref{eq_OutCoopSum1} and \eqref{eq_OuterSum2} with \eqref{eq_OutCoopSum2} respectively:
\begin{align*}
\eqref{eq_OutCoopSum1} - \eqref{eq_OuterSum1} &\le \mcal{C}\lp2(\SNR_1+\SNR_2)\rp - \mcal{C}\lp\SNR_1+\SNR_2\rp \le \frac{1}{2}\\
\eqref{eq_OutCoopSum2} - \eqref{eq_OuterSum2} &\le \frac{1}{2}+\C_{21} \le 1.
\end{align*}
Using the same argument in Section \ref{app_PfConstGapwoCoop}, $\eqref{eq_OutCoopSum1} - \eqref{eq_InnCoopSum} \le 1.5$ bits, and $\eqref{eq_OutCoopSum2} - \eqref{eq_InnCoopSum} \le 2$ bits.

The inner bound \eqref{eq_InnCoopR2} becomes $\frac{1}{2}\log^+\lp \frac{1+2\SNR_2}{2}\rp + \C_{21}$, and hence $\eqref{eq_OutCoopR2} - \eqref{eq_InnCoopR2} \le 0.5$ bits.

{\flushleft \it Case $\C_{21} \ge \frac{1}{2}$}:\par
We shall distinguish into three cases.
{\flushleft 1) $\Theta_{\mfrak{C}} = \lp N_o+2P_2\rp\lp2^{2\C_{21}} -2\rp$}: In this case $r_{21}=\C_{21}$. The inner bound \eqref{eq_InnCoopR2} becomes
\begin{align*}
&\frac{1}{2}\log^+\lp \frac{1+2\SNR_2}{2}\rp + \lp \frac{1}{2}\log\lp 2^{2\C_{21}} -1\rp-\frac{1}{2} \rp^+\\
&\ge \frac{1}{2}\log\lp1+2\SNR_2\rp - 1 + \frac{1}{2}\log\lp2^{2\C_{21}}-1\rp\\
&\ge \frac{1}{2}\log\lp1+\SNR_2\rp - 1 + \C_{21} - 1/2 =  \eqref{eq_OutCoopR2} - 3/2.
\end{align*}
Hence the gap is upper bounded by $1.5$ bits.

Since $Q_2\ge \Theta_{\mfrak{C}}$, the inner bound \eqref{eq_InnCoopSum} is lower bounded by
\begin{align*}
&\frac{1}{2}\log\lp1+2\SNR_2\rp - 3/2 + \C_{21}\\
& + \underset{(**)}{\underbrace{\frac{1}{2}\log\lp \frac{1+\INR_2+\SNR_1+\SNR_2}{1+2\SNR_2+2\INR_2}\rp}}.
\end{align*}
\begin{itemize}
\item
If $2\INR_2 \ge 1+2\SNR_2$, the term $(**)$ is lower bounded by $\frac{1}{2}\log\lp \frac{1+\INR_2+\SNR_1+\SNR_2}{4\INR_2}\rp$, and hence the gap to outer bound \eqref{eq_OutCoopSum2} is upper bounded by $3/2+\frac{1}{2}\log8=3$ bits.

\item
If $2\INR_2 < 1+2\SNR_2$, the term $(**)$ is lower bounded by $\frac{1}{2}\log\lp \frac{1+\SNR_1+\SNR_2}{2\lp1+2\SNR_2\rp}\rp$, and hence the gap to outer bound \eqref{eq_OutCoopSum1} is upper bounded by $3/2+\frac{1}{2}\log4=2.5$ bits.
\end{itemize}

{\flushleft 2) $\Theta_{\mfrak{C}} = Q_2$}: In this case $r_{21} = \mcal{C}\lp1+\frac{\INR_2}{1+2\SNR_2}\rp$. The inner bound \eqref{eq_InnCoopR2} becomes 
\begin{align*}
&\frac{1}{2}\log^+\lp \frac{1+2\SNR_2}{2}\rp + \lp \mcal{C}\lp \frac{\INR_2}{1+2\SNR_2}\rp-\frac{1}{2}\rp^+\\&\quad + \C_{21}-\mcal{C}\lp1+\frac{\INR_2}{1+2\SNR_2}\rp\\
&\ge \frac{1}{2}\log\lp1+\SNR_2\rp + \C_{21} -3/2 =  \eqref{eq_OutCoopR2} - 3/2.
\end{align*}
Hence the gap is upper bounded by $1.5$ bits. Analysis of the gap from the inner bound \eqref{eq_InnCoopSum} to outer bounds \eqref{eq_OutCoopSum1} \eqref{eq_OutCoopSum2} follows the same argument as that in 1).

{\flushleft 3) $\Theta_{\mfrak{C}} = P_1-P_2$}: The inner bound \eqref{eq_InnCoopR2} becomes inactive since it is greater than the inner bound \eqref{eq_InnCoopSum}, which is lower bounded by
\begin{align*}
&\frac{1}{2}\log\lp \frac{1+2\SNR_2}{2}\rp +  \frac{1}{2}\log\lp \frac{1+\SNR_1+\SNR_2}{1+2\SNR_2}\rp - \frac{1}{2}\\
&= \frac{1}{2}\log\lp1+\SNR_1+\SNR_2\rp - 1.
\end{align*}
It is within $1.5$ bits to the outer bound \eqref{eq_OutCoopSum1}.

Combining the above analysis, we complete the proof of Theorem \ref{thm_ConstGapCoop}.

\section{Proof of Claim \ref{claim_1}}\label{app_PfClaim1}
Use the random binning scheme in \cite{GelfandPinsker_80}, the rate $I\lp U;Y\rp - I\lp U;S\rp$ is achievable. We choose $U=X+\alpha S$, $X\sim\mcal{N}(0,P)$ and independent of $S$, and $\alpha = \frac{P}{P+N_Z}$, as in \cite{Costa_83}. Therefore, $X+S = \E\lb X+S| U\rb + Z' = rU + Z'$, where $Z'$ is Gaussian and independent of $\{U, Z'\}$. Rewrite $Y=rU+Z'+Z$.

Note that $I\lp U;S\rp$ is a fixed number that does not depend on the distribution of $Z$. We focus on lower bounding $I\lp U;Y\rp$. Use the argument that Gaussian noise is the worst case noise in an additive noise channel \cite{DiggaviCover_01} and note the $Y$ is the channel output with input $U$ and additive noise $Z'+Z$, we conclude the $I\lp U;Y\rp$ is minimized when $Z$ is Gaussian. Combining the classical dirty-paper coding result \cite{Costa_83} we complete the proof.

\end{document}